\def\year{2021}\relax
\documentclass[letterpaper]{article} 
\usepackage{aaai21}  
\usepackage{times}  
\usepackage{helvet} 
\usepackage{courier}  
\usepackage[hyphens]{url}  
\usepackage{graphicx} 
\usepackage{natbib}  
\usepackage{caption} 
\usepackage{amsmath,amssymb,xcolor,graphicx,enumitem,mathtools,booktabs}
\usepackage{microtype}
\usepackage{xr}

\usepackage{amsthm}
\usepackage{cleveref}
\newtheorem{theorem}{Theorem}
\newtheorem{lemma}[theorem]{Lemma}

\newtheorem{corollary}[theorem]{Corollary}
\newtheorem{definition}{Definition}

\usepackage{mathtools,eurosym}
\usepackage{thm-restate}
\DeclarePairedDelimiterX{\inp}[2]{\langle}{\rangle}{#1, #2}
\newcommand{\cov}{\textup{cover}}
\newcommand{\res}{\text{resid}}

\renewcommand*{\le}{\leqslant}
\renewcommand*{\ge}{\geqslant}
\renewcommand*{\leq}{\leqslant}
\renewcommand*{\geq}{\geqslant}
\renewcommand*{\epsilon}{\varepsilon}

\newcommand{\emdash}{\,---\,}

\newcounter{note}[section]

\DeclareMathOperator*{\argmax}{arg\,max}
\newcommand{\sw}{\ensuremath{\textup{sw}}}

\newcommand{\E}{\mathop{{}\mathbb{E}}}

\DeclareMathOperator*{\OPT}{\textup{OPT}} 
\newcommand{\poly}{\text{poly}}

\newcommand{\eps}{\epsilon}

\newcommand{\mcI}{\mathcal{I}}

\newcommand{\mcP}{\mathcal{P}}

\newcommand{\mcW}{\mathcal{W}}

\makeatletter
\newcommand*{\addFileDependency}[1]{
	\typeout{(#1)}
	\@addtofilelist{#1}
	\IfFileExists{#1}{}{\typeout{No file #1.}}
}
\makeatother

\allowdisplaybreaks

\nocopyright

\title{District-Fair Participatory Budgeting}
\author{
	D. Ellis Hershkowitz,\textsuperscript{\rm 1} Anson Kahng,\textsuperscript{\rm 1} Dominik Peters,\textsuperscript{\rm 2} Ariel D. Procaccia\textsuperscript{\rm 2} \\
}
\affiliations{
	\textsuperscript{\rm 1} Carnegie Mellon University \\
	\textsuperscript{\rm 2} Harvard University \\
	dhershko@cs.cmu.edu, akahng@cs.cmu.edu, dpeters@seas.harvard.edu, arielpro@seas.harvard.edu
}

\begin{document}

\maketitle

\begin{abstract}

    Participatory budgeting is a method used by city governments to select public projects to fund based on residents' votes. Many cities use participatory budgeting at a district level. Typically, a budget is divided among districts proportionally to their population, and each district holds an election over local projects and then uses its budget to fund the projects most preferred by its voters. However, district-level participatory budgeting can yield poor social welfare because it does not necessarily fund projects supported across multiple districts. On the other hand, decision making that only takes global social welfare into account can be unfair to districts: A social-welfare-maximizing solution might not fund any of the projects preferred by a district, despite the fact that its constituents pay taxes to the city. Thus, we study how to fairly maximize social welfare in a participatory budgeting setting with a single city-wide election. We propose a notion of fairness that guarantees each district at least as much welfare as it would have received in a district-level election. We show that, although optimizing social welfare subject to this notion of fairness is NP-hard, we can efficiently construct a lottery over welfare-optimal outcomes that is fair in expectation. Moreover, we show that, when we are allowed to slightly relax fairness, we can efficiently compute a fair solution that is welfare-maximizing, but which may overspend the budget.
    

\end{abstract}

	\section{Introduction}

Participatory budgeting is a democratic approach to the allocation of public funds. In the participatory budgeting paradigm, city governments fund public projects based on constituents' votes.
In contrast to budget committees, which operate behind closed doors, participatory budgeting promises to directly take the voices of the community into account. Since 2014, Paris has allocated more than \euro 100 million per year using constituents' votes. Many other cities around the globe\emdash including Porto Alegre, New York City, Boston, Chicago, San Francisco, Lisbon, Madrid, Seoul, Chengdu, and Toronto\emdash employ participatory budgeting \citep{Cabannes, cabannes2014contribution, aziz2020participatory}.

Typically, participatory budgeting is used at a district-level. Each district of the city is allotted a budget proportional to its size. Constituents living in a given district vote on projects such as park, road or school improvements local to the district, using some version of approval voting. Then, the district's budget is spent according to these votes. For instance, in Paris a participatory budget is split between 20 districts (a.k.a.\ arrondissements), constituents vote and then each district runs a greedy algorithm to maximize the total social welfare\emdash i.e., the total number of votes\emdash of the funded projects.\footnote{More specifically, projects are selected in descending order of vote count until the budget runs out.}

Having separate elections for each district leads to several problems. Foremost, projects that are not local to a single district cannot be accommodated. For this reason, Paris must run an additional election for city-wide projects. However, this splits the available budget for participatory budgeting between district-level and city-wide elections in an ad hoc manner, which is not informed by votes.\footnote{In 2016, this split in Paris was \euro 64.3 million for district elections and \euro 30 million for city-wide elections~\citep{cabannes2017participatory}.} Further, people may have interests in multiple districts, such as those who live and work in different districts. For this reason, Paris has to allow residents to choose the district in which they vote. Lastly, a project that only benefits voters at the edge of a district may receive a number of votes that is not proportional to the number of potential beneficiaries.

A simple solution to these problems is a single city-wide election. However, such a voting scheme may result in unfair outcomes. For instance, if votes are aggregated to maximize social welfare (i.e., as is presently done in Paris on the district level) then it is possible that some districts might have none of their preferred projects funded despite deserving a large proportion of the budget. Such outcomes are likely when some districts are much more populous than others, in which case projects local to small districts cannot gather sufficiently many votes.
Ideally, we would like a system that balances the tradeoff between social welfare and fairness without an arbitrary, pre-determined split between district-specific and city-wide funding. 
This motivates our central research question: 
\begin{quote}
	\emph{How can we maximize social welfare in a way that is fair to all districts?}
\end{quote}


Intuitively, a solution that is fair to all districts should somehow represent each districts' constituents. One way to formalize this intuition is to stipulate that no district should be able to obtain higher utility by purchasing projects with its proportional share of the budget. In particular, each district should receive at least as much utility as it would have received had it held a district-level election with its proportional share of the budget.
We call this guarantee \emph{district fairness}.\footnote{Our notion of district fairness can be thought of as a form of \emph{individual rationality} where every district is seen as an ``individual.''} A district-fair allocation of funds always exists, since an outcome obtained by holding separate district elections is district fair.
We aim to find district-fair outcomes that maximize social welfare. Such an outcome will be a Pareto-improvement on the status quo of district-level participatory budgeting, in the sense that each district's welfare has increased.

\paragraph{Our Results.}

In our model we think of (utilitarian) social welfare as induced by a given value assigned by each district to each project; our goal is to maximize the sum of these values over districts and selected projects. Note that this model captures the setting of approval votes, where each voter decides on a collection of projects to vote for; the social welfare of a district for a project would then be interpreted as the project's overall number of approvals from voters in that district. This observation is important because some variant of approval voting is used in most real-world participatory budgeting elections, including in Paris. 

We also assume that each district is endowed with an arbitrary fraction of the total budget. Clearly this captures, as a special case, the common setting where the endowment of each district is proportional to its size. Moreover, the reasoning behind the existence of district-fair outcomes immediately applies to the more general setting. 

%

We first show that it is NP-complete to compute an allocation that is welfare-maximizing subject to district fairness. This result holds even for the case of approval votes and proportional budgets, and therefore the generality of our model only strengthens our positive (algorithmic) results without weakening the main negative (hardness) result. We also show that the natural linear program (LP) formulation of the problem has an unbounded integrality gap. Since participatory budgeting elections can be large\emdash hundreds of projects are proposed and hundreds of thousands of votes are cast in Paris\emdash computational complexity can become a problem in practice. Thus, we seek polynomial-time solutions with reasonable approximation guarantees.

There are several ways one might relax our problem or trade-off between parameters in our problem. In this work, we design polynomial-time algorithms that work when we relax or approximate some of the following: (1) the achieved social welfare; (2) the spent budget; (3) the fairness of the solution; and (4) the absence of randomization.

We first relax (4) by considering distributions over outcomes, a.k.a.\ ``lotteries''. We show that using a multiplicative-weights-type algorithm, one can efficiently find a lottery that guarantees budget feasibility (ex post), optimum social welfare (ex post), and district-fairness in expectation up to an $\epsilon$ (ex ante). Since the fairness guarantee only holds in expectation, some districts may be underserved once the lottery is realized. However, since participatory budgeting typically happens repeatedly (e.g., annually), such districts could be compensated in the next election, for example by increasing their share of the budget in the next year.

We next consider what sort of deterministic guarantees are achievable. To this end, we show how to use techniques from submodular optimization to find an outcome that is district fair ``up to one project'' and which achieves optimum social welfare with the caveat that the outcome may need to spend $64.7\% $ more money than was originally budgeted. We also give a randomized algorithm with the same guarantees but which overshoots the budget by only a $1/e$ $(\approx 37\%)$ fraction with high probability. 
Additionally, as a corollary of these results, we give both deterministic and randomized algorithms that achieve weaker utility and fairness guarantees but do not overspend the available budget.

\paragraph{Related Work.}

The social choice literature on participatory budgeting has both studied the voting rules used in practice, and designed original voting schemes. \citet{goel2019knapsack} study knapsack voting, used for example in Madrid \citep{cabannes2014contribution}, where voters cannot approve more projects than fit into the budget constraint. \citet{talmon2019framework} axiomatically study a variety of approval-based rules that maximize social welfare, both greedy and optimal ones.

The unit cost case (where all projects have the same cost) is best-studied, as \emph{multi-winner} or \emph{committee} elections \citep{FSST17a}. For example, this setting models the election of a parliament.
A main focus of that literature is the computational complexity of the winner determination of various voting rules. More relevant for our purposes are fairness axioms used in this setting. The most prominent such axioms are variants of \emph{justified representation} \citep{justifiedRepresentation}. These axioms are formulated for approval votes, and require that arbitrary subgroups of the electorate need to be represented in the outcome if they are \emph{cohesive}, in the sense that there are a sufficient number of projects that are approved by every member of the subgroup. Several voting rules are known to satisfy these conditions, including Phragm\'en's rule and Thiele's Proportional Approval Voting \citep{Janson16arxiv,pjr17,BrillFJL17phragmen,AEHLSS18}. By contrast, district-fairness gives guarantees to a specific selection of subgroups (i.e., disjoint districts) but does not require these groups to be cohesive.

A very strong fairness axiom that is sometimes discussed in the context of committee elections and participatory budgeting is the \emph{core} \citep{FGM16a,justifiedRepresentation,FMS18}. It insists that every subgroup (or \emph{coalition}) must be represented (in the sense that it should not be possible for the subgroup to propose an alternative use of their proportional share of the budget that each group member prefers to the chosen outcome), without a cohesiveness requirement. For approval-based elections, it is a major open question whether there always exists a core outcome. For general additive utilities, there are instance where no core outcome exists \citep{FMS18}, but several researchers have proved the existence of approximations to the core \citep{jiang2020approximately,FMS18,cheng2019group,peters2020proportionality}. A district-fair outcome is, in a sense, in the core: no subgroup which coincides with a district can block the outcome. Thus, our work shows that for general utilities, a core-like outcome exists if we only allow a specific collection of (disjoint) coalitions to block. 


The problem of \emph{knapsack sharing} \citep{brown1979knapsack} has a similar motivation to our problem. The knapsack sharing problem supposes that the \emph{projects} are separated into districts (instead of, in our case, the voters), and each project comes with a cost and a value. The aim is to find a budget-feasible set of projects that maximize the minimum total value of the projects in a district. Note that in this formulation all districts are treated equally (there is no weighting by district population) and that there is no notion of the value of a project to a specific district. The literature contains a variety of algorithms for solving this NP-hard problem \citep[e.g.,][]{yamada1997heuristic,yamada1998some,hifi2005exact,fujimoto2006exact}.

\section{Formal Problem, Notation and Definitions}\label{sec:formalProb}

Formally, the setting we consider is as follows. We are given a budget $b \in \mathbb{Z}_{\geq 1}$. There are $m$ possible projects $\mcP = \{x_1, \ldots, x_m\}$ with associated nonnegative costs $c: \mcP \to \mathbb{Z}_{\geq 0}$. We refer to a subset $W \subseteq \mcP$ as an \emph{outcome}. The cost of an outcome $W$ is $c(W) \coloneqq \sum_{x_j \in W} c(x_j)$. We say that a subset $W$ is \emph{budget-feasible} if $c(W) \leq b$.

There are $k$ districts $d_1, \ldots, d_k$.
The \emph{social welfare} (or \emph{utility}) that project $x_j$ provides to district $d_i$ is $\sw_i(x_j) \in \mathbb{Z}_{\geq 0}$. We assume that utilities are additive; i.e., the utility that an outcome $W \subseteq \mcP$ provides to district $d_i$ is $\sw_i(W) := \sum_{x_j \in W} \sw_i(x_j)$. Furthermore, the total social welfare of $W \subseteq \mcP$ is $\sw(W) \coloneqq \sum_{i \in [k]} \sw_i(W)$.

Throughout this work we assume that $\sw(x_j)$ and $c(x_j)$ are both  $\poly(k,m)$ for each $j$. (A function $f$ is $\poly(x,y)$ if there exists a $k \geq 0$ such that $f = O((xy)^k)$.) We can relax this assumption using well-known bucketing techniques at the cost of an arbitrarily small $\epsilon$ in the guarantees of our algorithms. See the fully polynomial time approximation scheme for the knapsack problem \citep{chekuri2005polynomial} for an example of this technique.

To model the participatory budgeting setting, we assume that each district deserves some portion of the budget and, in turn, deserves at least the utility it could achieve if it spent its budget on its most preferred projects. Specifically, each district $d_i$ deserves some budget $b_i \geq 0$ where $\sum_{i} b_i = b$. District $d_i$ deserves utility $f_i := \sw_i(W_i)$, where $W_i := \argmax_{W : c(W) \leq b_i} \sw_i(W)$ is $d_i$'s favorite outcome costing at most $b_i$.
\begin{definition}[District-Fair Outcome]
    We say that an outcome $W$ is district-fair (DF) if $\sw_i(W) \geq f_i$ for all $i$.
\end{definition}

Computing $f_i$ is precisely an instance of the knapsack problem; by our assumption that utilities and costs are polynomially bounded, this knapsack instance is solvable in polynomial time \citep{chekuri2005polynomial}. Thus, we will assume $f_i$ is known.

Note that the outcome $\bigcup_i W_i$ is both budget-feasible and district-fair, so an outcome with both properties always exists.
Our goal is to find a budget-feasible and district-fair outcome $W$ which maximizes social welfare $\sw(W)$.
We call our problem \emph{district-fair welfare maximization}. 
Throughout this paper, we let $W^* := \argmax_W \sw(W)$ be some optimal solution, where the argmax is taken over budget-feasible and district-fair solutions. Similarly, we let $\OPT := \sw(W^*)$.

We consider two relaxations of district fairness. The first relaxation extends the concept to \emph{lotteries} over outcomes. We require that each district only needs to be approximately satisfied in expectation.
We give an efficient algorithm to compute optimal district-fair lotteries in \Cref{sec:lottery}.

\begin{definition}[$\epsilon$-District-Fair Lottery]
	Given $\epsilon > 0$, we say that a probability distribution $\mcW$ over outcomes of cost at most~$b$ is an $\epsilon$-district-fair ($\epsilon$-DF) lottery if  $E_{W \sim \mcW}[\sw_i(W)] \geq f_i - \epsilon$ for every district $d_i$.
\end{definition}

The second relaxation is \emph{district-fairness up to one good} (DF1). Intuitively, an allocation is DF1 if each district would be satisfied if one additional project was funded.

\begin{definition}[DF1]
	An outcome $W$ is \emph{DF1} if for every $d_i$, 
	\[\sw_i(W) + \max_{x_j \in (\mcP \setminus W)} \sw_i(x_j) \ge f_i.\]
\end{definition}

DF1 is inspired by the well-studied notion of EF1 (envy-freeness up to one good) from the private goods setting \citep{budish2011combinatorial}.
This relaxation is mild, and unlike relaxations that require district-fairness to hold on average over districts, it is a \emph{uniform} relaxation which provides guarantees for all districts.
We study DF1 outcomes in \Cref{sec:DF1}.

\section{NP-Hardness}

Our first result shows that the problem of optimizing social welfare subject to district-fairness is NP-hard even in the restricted setting of approval votes (i.e., voters provide binary yes/no opinions over projects) and budgets proportional to district sizes. In fact, our problem remains NP-hard in this restricted setting even when each district contains only one voter and projects have unit costs.

We reduce from exact 3-cover (X3C), which is known to be NP-hard \citep{garey1979computers}. 
The idea of our reduction is as follows. Given an instance of X3C, we define a district for each of the elements in the universe, and then add a large amount of dummy districts. We then define a project for each set in our problem instance which gives one utility to the districts corresponding to the elements which it covers. We also define a large set of dummy projects that are approved by all dummy districts. We then ask whether there exists a district-fair outcome that attains high social welfare. An optimal solution for our district-fair welfare maximization problem, then, will first try to solve the X3C instance as efficiently as possible so that it can spend as much of its budget as possible on high-utility dummy projects. We formalize this idea in the following proof. 

\begin{restatable}{theorem}{NPHard}
	It is NP-complete to decide, given an instance of district-fair welfare maximization and an integer $M$, whether there exists a budget-feasible and district-fair outcome $W$ such that $\sw(W) \ge M$. NP-hardness holds even in the restricted setting of approval votes and budgets proportional to district sizes, and when each district contains one voter and all projects have unit cost.
\end{restatable}
\begin{proof}
    The stated problem is trivially in NP. For NP-hardness we reduce from X3C. In an instance of X3C, we are given a universe $U = \{e_1, \dots, e_{3n}\}$ and a collection $\{S_1, \dots, S_m\}$ of 3-element subsets of $U$. It is a ``yes''-instance if there exists a selection $S_{j_1}, \dots, S_{j_n}$ such that $S_{j_1} \cup \dots \cup S_{j_n} = U$. 
	
	Given an instance of X3C, we construct an instance of our problem as follows. Let $M = 3mn + 1$. We have $3n + M$ districts, $D \cup D'$. Let $D = \{d_1, \dots, d_{3n}\}$, where each $d_i$ in $D$ corresponds to element $e_i$. Additionally, let $D' = \{d_{3n+1}, \ldots d_{3n + M}\}$, where each $d_i \in D'$ is a dummy district. We have $m + 2n + M$ projects, $X \cup X'$. Let $X = \{x_1, \dots, x_m\}$, where $x_j \in X$ corresponds to set $S_j$, and let $X' = \{x_{m+1}, \dots, x_{m+2n+M} \}$, where each $x_j \in M'$ is a dummy project. Utilities are as follows: every dummy district approves every dummy project, so $\sw_{i}(x_j) = 1$ for each $i \geq 3n+1$ and $x_j \in X'$. Also, each non-dummy district approves of non-dummy sets to reflect the structure of the X3C instance: that is, for each $i \leq 3n$ we have $\sw_{i}(x_j) = 1$ if $x_j \in X$ and $e_i \in S_j$. All other utilities are 0: that is, $\sw_i(x_j)=0$ for all other $i$ and $j$. Each project has cost 1, and our budget is $b = 3n+M$. We assume all districts contain 1 voter, so $b_{i} = 1$ for every district $d_i$. Clearly, $f_{i} = 1$ for each $i$. We ask whether there exists a district fair committee with social welfare at least $3n + (2n+M)M$.
	
	If there exists a solution $S_{j_1}, \dots, S_{j_n}$ to the X3C instance, then $W = \{x_{j_1}, \dots, x_{j_n}\} \cup X'$ is an outcome with cost $n + (2n +M) = 3n+M = b$. Clearly, $W$ is district-fair, and its social welfare is $3n + (2n+M)M$, so this is a ``yes''-instance for the district-fair welfare-maximization problem.
	
	Conversely suppose that there exists a district-fair budget-feasible outcome $W$ with social welfare at least $3n + (2n+M)M$. Note that all projects in $X$ together give overall welfare at most $3mn < M$. Thus, we must have $X' \subseteq W$ since otherwise the total welfare of $W$ is less than $(2n+M)M$. Hence $|X \cap W| \le n$. By district-fairness, for each $i = 1, \dots, 3n$, there must be some $x_j \in W$ such that $e_i \in S_j$. These two facts together imply that $\{S_j : x_j \in W\}$ is a solution to the X3C instance.
\end{proof}

This NP-hardness result holds even if each district consists of a single voter and all projects have unit cost. As we show in 
\Cref{app:halfapx} in the supplementary material, 
this special case admits a polynomial-time $\frac{1}{2}$-approximation.
Our algorithm is based on a greedy algorithm and a combinatorial argument which ``matches away'' high utility goods of the optimal solution.
One might hope to achieve an approximation result for the general case. A natural approach would be to round the optimal solution to the LP relaxation of the natural ILP formulation of our problem. However, a simple example in 
\Cref{app:integralitygap} in the supplementary material
shows that the integrality gap of that formulation is unboundedly large, so this approach will not work. 

\section{Optimal District-Fair Lottery}
\label{sec:lottery}

In this section, we allow randomness and consider lotteries over outcomes. 
Our main result for the lottery setting is an $\epsilon$-DF lottery which always achieves the optimal social welfare subject to district fairness. The welfare guarantee is ex post, so that every outcome in the lottery's support achieves optimal welfare. For the remainder of this section we let $\eps > 0$ refer to the $\eps$ in the $\epsilon$-DF definition.
\begin{restatable}{theorem}{lottery}
\label{thm:lottery}
	There is an algorithm which, in $\poly\left(m,k, \frac{1}{\epsilon}\right)$ time, returns an $\epsilon$-DF lottery $\mcW$ such that for all outcomes $W$ in the support of $\mcW$, we have $\sw(W) \geq \OPT$.
\end{restatable}

The intuition for our algorithm is as follows. We begin by showing that our problem is polynomial-time solvable if the number of districts $k$ is constant. Such an algorithm is useful because we can artificially make the number of districts constant by convexly combining all districts into a single district $\tilde{d}$. We can, then, compute as our solution a utility-optimal outcome $W$ which is fair for $\tilde{d}$ but not necessarily fair for each $d_i$ individually. However, we can bias our solution to try and satisfy fairness for certain districts by increasing the weights of these districts in our convex combination. Thus, if $W$ is not fair for $d_i$, we might naturally increase the proportional share of $d_i$ in the convex combination and recompute $W$ in the hopes that the new outcome we compute will be fair for $d_i$. We obtain our lottery by repeatedly increasing the weight of districts that do not have their fairness constraint satisfied, and then take a uniform distribution over the resulting outcomes.

Turning to the proof, we begin by describing how to solve our problem in polynomial time when $k$ is a constant. Our algorithm will solve the natural dynamic program (DP). Specifically, consider the true/false value $R(\sw^{(1)}, \ldots, \sw^{(k)},j,b)$ which is the answer to the question, ``Does there exists an outcome of cost at most $b$ using projects $x_1, x_2, \ldots, x_{j}$ wherein district $d_i$ achieves social welfare at least $\sw^{(i)}$?'' If the answer to this question is yes, then either the desired utilities are possible with the stated budget without using $x_j$ or there is an outcome which uses at most $b - c(x_j)$ budget that doesn't use $x_j$ in which every district gets at least its specified utility minus how much it values $x_j$. Thus, $R(\sw^{(1)}, \ldots, \sw^{(k)},j,b)$ is true if and only if either $R(\sw^{(1)}, \ldots, \sw^{(k)},j-1,b)$ is true or $ R(\sw^{(1)} - \sw_1(x_j), \ldots, \sw^{(k)} - \sw_k(x_j),j-1,b-c(x_j))$ is true, giving us a definition by recurrence.

By our assumption that all costs and utilities are polynomially bounded, we can easily solve the dynamic program (DP) for the above recurrence, giving the following result.
\begin{lemma}\label{lem:exactConstantk}
	There is an algorithm that finds a budget-feasible district-fair outcome $W$ with $\sw(W) = \OPT$ in $m^{O(k)}$ time.
\end{lemma}
\begin{proof}
	Our algorithm simply fills in the DP table and returns the outcome corresponding to the entry in our DP table which is true, satisfies $\sw^{(i)} \geq f_i$ for all $i$ and which maximizes $\sum_i \sw^{(i)}$. The recurrence is correct by the above reasoning.
	
	To see why we can fill in the DP table in the stated time, note that we can trivially solve our base case, $R(\sw^{(1)}, \ldots, \sw^{(k)},j,1)$, for each $j$ and possible value for each $\sw^{(i)}$ in polynomial time. Since $\max_{i,j}\sw_i(x_j)$ is polynomially bounded in $m$, we need only check polynomially-many in $m$ values for each $\sw^{(i)}$. Lastly, since $j$ and $b$ are bounded by a polynomial in $m$, we conclude that our DP table has $m^{O(k)}$ entries, giving the desired runtime.
\end{proof}

We now describe our multiplicative-weights-type algorithm to produce our lottery using the above algorithm.\footnote{We will only need to invoke the above algorithm for the case $k=1$. This amounts to solving the knapsack problem with a single covering constraint, which to our knowledge is not one of the standard variants of the knapsack problem.} We let $w_i^{(t)} \ge 0$ be the ``weight'' of district $i$ in iteration $t$ and let $w^{(t)} := \sum_i w_i^{(t)}$ be the total weight in iteration $t$. Initially, our weights are uniform: $w_i^{(1)} = 1$ for all $i$. 

For any iteration $t$ and district $d_i$ we let $p_i^{(t)} := \frac{w_i^{(t)}}{w^{(t)}}$ be the proportion of the weight that district $i$ has in iteration $t$. These $p_i^{(t)}$ will induce our convex combination over districts; in particular we let $\tilde{d}^{(t)}$ be a district which values project $x_j$ to extent $\tilde{\sw}^{(t)}(x_j) := \sum_i p_i^{(t)} \cdot \sw_i(x_j)$ and which deserves $\tilde{f}^{(t)} := \sum_i p_i^{(t)} \cdot f_i$ utility. Also, let $\sw_{\max}$ be the maximum welfare of an outcome.

With the above notation in hand, we can give our instantiation of multiplicative weights where $T := \frac{4 \ln k}{\epsilon^2} \cdot \sw_{\max}^2$ is the number of iterations of our algorithm.

\begin{enumerate}
	\item For all iterations $t \in [T]$:
	\begin{enumerate}
		\item Let $W_t$ be an outcome that maximizes $\sw(W_t)$ subject to $\tilde{\sw}^{(t)}(W_t) \geq \tilde{f}^{(t)}$ and $c(W_t) \leq b$. We can compute $W_t$ using \Cref{lem:exactConstantk}.
		\item Let $m_i^{(t)} := \sw_i(W_t) -  f_i$ be our ``mistakes'', indicating how far off a district was from getting what it deserved.
		\item Update weights: $w_i^{(t+1)} \gets w_i^{(t)} \cdot \exp(- \epsilon m_i^{(t)})$.
	\end{enumerate}
	\item Return lottery $\mcW$, the uniform distribution over $\{W_t\}_t$.
\end{enumerate}

We now restate the usual multiplicative weights guarantee in terms of our algorithm. This lemma guarantees that, on average, the multiplicative weights strategy is competitive with the best ``expert.'' In the following $\inp{p^{(t)}}{m^{(t)}} := \sum_i p^{(t)}_i \cdot m_i^{(t)}$ is the usual inner product.

\begin{lemma}[\citealp{arora2012multiplicative}]\label{lem:mwLemma}
	For all $i$ we have
	\begin{align*}
	\frac{1}{T} \sum_{t \leq T} \inp{p^{(t)}}{m^{(t)}} \leq \epsilon + \frac{1}{T} \sum_{t \leq T} m_i^{(t)}.
	\end{align*}
\end{lemma}

We can use this lemma to show the desired guarantees.

\begin{proof}[Proof of Theorem~\ref{thm:lottery}]
	We use the algorithm described above.
	
	Our algorithm is polynomial time since it runs for polynomially-many iterations and in each iteration we compute a solution for a problem on only one district which is solvable in polynomial time by \Cref{lem:exactConstantk}. Also, note that by \Cref{lem:exactConstantk} we know that $c(W_t) \leq b$ for all $t$, so all outcomes in the lottery are budget-feasible.
	
	We now argue that the above lottery is utility-optimal. Fix an iteration $t$. Notice that since $W^*$ is fair for all districts then it is fair for $\tilde{d}^{(t)}$. In particular, 
	\begin{align*}
	\tilde{\sw}^{(t)}(W^*) = \sum_i p_i \cdot \sw_i(W^*)	
	\geq \sum_i p_i f_i
	= \tilde{f}^{(t)}
	\end{align*}
	Thus, $W^*$ is a budget-feasible solution for the problem of finding a max-utility outcome which is fair for $\tilde{d}^{(t)}$. Thus, $\sw(W_t)$ can only be larger than $\sw(W^*)$, meaning that $\sw(W_t) \geq \OPT$.
	
	We now argue that the above lottery is $\epsilon$-DF in expectation. Fix a district $d_i$. By Lemma~\ref{lem:mwLemma} we know that 
	\begin{align}\label{eq:mwGuar}
	\frac{1}{T} \sum_{t \leq T} \inp{p^{(t)}}{m^{(t)}} \leq \epsilon + \frac{1}{T} \sum_{t \leq T} m_i^{(t)}.
	\end{align}
	
	Now notice that by definition of $m_i^{(t)}$ and since our lottery is uniform over all $W_t$ we know that the right-hand-side of \Cref{eq:mwGuar} is
	\begin{align*}
	\epsilon + \frac{1}{T} \sum_{t \leq T} m_i^{(t)} &= \epsilon +  \frac{1}{T} \sum_{t} (\sw_i(W_t) - f_i)\\
	&= \epsilon - f_i + \frac{1}{T}\sum_t \sw_i(W_t)\\
	&= \epsilon -f_i + \E_{W \sim \mcW}[\sw_i(W)]
	\end{align*}
	
	Thus, to show that $ f_i - \epsilon \leq \E_{W \sim \mcW}[\sw_i(W)]$, it suffices to show that the left-hand side of \Cref{eq:mwGuar} is at least $0$. That is, we must show $0 \leq \frac{1}{T} \sum_{t \leq T} \inp{p^{(t)}}{m^{(t)}}$. However, this amounts to simply showing that $W_t$ is fair for $\tilde{d}^{(t)}$; in particular, we have that the left-hand-side is
	\begin{align*}
	\frac{1}{T} \sum_{t \leq T} \inp{p^{(t)}}{m^{(t)}} &= \frac{1}{T} \sum_{t \leq T} \sum_i p^{(t)}_i \cdot (\sw_i(W_t) - f_i)\\
	&= \frac{1}{T} \sum_{t \leq T} \tilde{\sw}^{(t)}(W_t) - {\tilde{f}^{(t)}}.
	\end{align*}
	It holds that $\tilde{sw}^{(t)}(W_t) - \tilde{f}^{(t)} \geq 0$ since we always choose a solution which is fair for $\tilde{d}^{(t)}$, and so we conclude that the left-hand-side of \Cref{eq:mwGuar} is at least $0$.
\end{proof}

\section{Optimal DF1 Outcome with Extra Budget}
\label{sec:DF1}

We now study how well we can do if we allow ourselves to overspend the available budget.
Certainly it is possible to achieve district fairness and optimal fairness-constrained utility $\OPT$ if the algorithm can spend \emph{double} the available budget: we can compute an outcome $W_1$ with $c(W_1) \le b$ that is welfare-maximizing without attempting to satisfy district-fairness, and we can compute some outcome $W_2$ with $c(W_2) \le b$ that is district-fair (see \Cref{sec:formalProb}); then $W_1 \cup W_2$ satisfies district fairness and we clearly have $c(W_1 \cup W_2) \le 2b$ and $\sw(W_1\cup W_2) \ge \OPT$.
In this section, we show that we can find a solution that requires less than twice the budget, if we slightly relax the district fairness requirement to DF1.
Our main result for the DF1 setting shows that, under DF1 fairness, there is a deterministic algorithm which achieves DF1 and optimal social welfare if one overspends a $0.647$ fraction of the budget. 

\begin{restatable}{theorem}{Fonedet}
	\label{thm:DF1det}
	For any constant $\eps > 0$, there is a $\poly(m,k)$-time algorithm which, given an instance of district-fair welfare maximization, returns an outcome $W$ such that $W$ is DF1, $c(w) \leq \left(1.647+ \epsilon\right) b$, and $\sw(W) \ge (1-\epsilon)\OPT$.
\end{restatable}

Overspending by 64.7\% is a worst-case result, and the algorithm may often overspend less. If the context does not permit any overspending, one can run the same algorithm with a reduced budget; then the output will be feasible for the true budget, yet will satisfy weaker fairness and social welfare guarantees.
More precisely, given an instance $\mcI$ and a multiplier $\beta < 1$, we define an instance $\mcI'(\beta)$, which is identical to $\mcI$ but in which each district $d_i$ contributes only $\beta \cdot b_i$ and thus deserves utility $f_i' \coloneqq \sw_i(W_i')$, where $W_i'$ is $d_i$'s favorite outcome which costs at most $\beta \cdot b_i$. Additionally, let $\OPT'(\beta)$ represent the maximum achievable social welfare over all district-fair solutions in $\mcI'$ using a budget of at most $b' \coloneqq \beta \cdot b$. Then, applying \Cref{thm:DF1det} to $\mcI'(\beta)$ results in an outcome which is DF1 and utility-optimal on this reduced instance and does not overspend the original budget $b$.

\begin{corollary}\label{cor:det}
    For any constant $\eps > 0$, there is a $\poly(m,k)$-time algorithm which, given an instance $\mcI$ of district-fair welfare maximization, returns an outcome $W$ such that $W$ is DF1 for $\mcI'(\frac{1}{1.647})$, $c(W) \leq \left(1 + \epsilon \right) b$, and $\sw(W) \ge (1 - \epsilon) \OPT'(\frac{1}{1.647})$.
\end{corollary}

Our result uses a submodular optimization as a subroutine. If one allows randomization in this subroutine, algorithms with better approximation ratios are known.
Thus, we can prove a similar theorem (and corollary) with a \emph{randomized} algorithm which achieves DF1 and optimal social welfare while overspending its budget by only a $\frac{1}{e} \approx .37$ fraction of the budget, with high probability (i.e., with probability $1 - \frac{1}{p(m,k)}$ where $p(m,k)$ is some polynomial in $m$ and $k$). We defer details of our randomized algorithm to \Cref{sec:randDF1} in the supplementary material.


In the remainder of this section, we will prove \Cref{thm:DF1det}. Our main tool is a notion of the ``coverage'' of a partial outcome. An outcome has high coverage if we do not need to spend much more money to make it district-fair. On a high level, our proof consists of two main steps. First, we show how to complete an outcome with good coverage into a DF1 outcome. Second, we will show how to frame the problem of finding a solution with good coverage and social welfare as a submodular maximization problem subject to linear constraints, allowing us to use a result by \citet{mizrachi2018tight}.

We begin by formalizing the coverage of a solution. Roughly, if we imagine that initially every district requires its portion of the budget for fairness, then fractional coverage captures how much less districts must spend to satisfy their own fairness constraints. Thus, if we imagine that our algorithm first spends its budget to satisfy fairness as efficiently as possible, and then spends the remainder of its budget on the highest utility projects, then the coverage of a collection of projects is roughly how much budget this collection ``frees up'' for the algorithm to spend on the highest utility projects. More formally, we define coverage by way of the notions of fractional outcomes and residual budget requirements.

\begin{definition}[fractional outcomes]
	A fractional outcome is a vector $p \in \mathbb{R}^m$ where $0 \leq p_j \leq 1$. We overload notation and let the social welfare of $p$ for district $d_i$ be $\sw_i(p):= \sum_j \sw_i(x_j) \cdot p_j$. Similarly the social welfare of $p$ is $\sum_i \sw_i(p)$. Lastly, we define the cost of $p$ as $\sum_j c(x_j)\cdot p_j$. 
\end{definition}

We now define the residual budget requirement of a district, given an outcome, which can be understood as the minimum amount of additional money that must be spent to satisfy the district, if fractional outcomes are allowed. 

\begin{definition}[$\res_i(W)$]
	The residual budget requirement of district $d_i$ given (integral) outcome $W$ is the minimum cost of a fractional outcome $p$ such that $\sw_i(W) + \sw_i(p) \geq f_i$ and $p_j = 0$ for all $x_j \in W$.
\end{definition}

We can now define the coverage of an outcome for a particular district $i$ in terms of the total amount of budget they deserve and their residual budget requirement.

\begin{definition}[$\cov_i(W)$]\label{dfn:covi}
	The coverage of an outcome $W$ for district $d_i$ is the difference between the amount of budget they deserve, $b_i$, and their residual budget requirement: $\cov_i(W) := b_i - \res_{i}(W)$.
\end{definition}

Lastly, we define the coverage of an outcome.

\begin{definition}[$\cov(W)$]\label{dfn:coverage}
	The overall coverage of an outcome $W$ is the sum over all districts $d_i$ of the coverage $W$ affords $d_i$: $\cov_i(W) := \sum_{i} \cov_i(W)$.
\end{definition}

Next, we establish a useful property of DF1 solutions. In particular, given a set of projects that achieves relatively good fairness on \emph{average}, we can then buy a small subset of projects that results in fairness up to one good for all districts.
In particular, given a collection of projects that covers a $1-\beta$ fraction of all fairness constraints, we can use at most an extra $\beta$ fraction of our budget in order to complete this to a DF1 solution. Moreover, this completion is quite intuitive: purchase all projects whose total coverage exceed their cost, until there are no such projects remaining. 

Formally, we state the following \emph{DF1 completion} lemma.

\begin{lemma}[DF1 Completion]
	\label{lem:df1completion}
	Given an outcome $W$ with $\cov(W) = b - r$, one can compute in polynomial time a set $W' \supseteq W$ such that $W'$ is DF1 and $c(W') \leq c(W) + r$.
\end{lemma}
\begin{proof}
	We first prove that for every non-DF1 outcome $W$, there exists a project that we can add to $W$ which increases its coverage by at least $c(x_j)$. Suppose that $W$ is an outcome that fails DF1, and let $d_i$ be a district such that $\sw_i(W) + \sw_i(x_j) < f_i$ for all $x_j \not\in W$. Let $p$ be the fractional outcome witnessing  $\res_{i}(W)$; thus $\sw_i(W) + \sw_i(p) \ge f_i$. We may assume without loss of generality that all but at most one project is integral in $p_j$ (because there is always some optimal $p$ with this property by additivity of $\sw_i$). Since $W$ fails DF1 for $d_i$, there is some $x_j \not \in W$ such that $p(x_j) = 1$. Then $\res_i(W \cup \{x_j\}) = \res_i(W) - c(x_j)$ (witnessed by the fractional outcome obtained from $p$ by removing $x_j$ from it). Thus, from definitions, $\cov_i(W \cup \{x_j\}) = \cov_i(W) + c(x_j)$, and hence $\cov(W \cup \{x_j\}) \ge \cov(W) + c(x_j)$.
	
	Now suppose we are given an outcome $W$ with $\cov(W) = b - r$, which fails DF1. We can identify a project $x_j$ as above, add it to $W$, and increase the coverage by at least $c(x_j)$. We repeat this until the outcome is DF1. This process must stop, since at each step the coverage increases by $c(x_j)$ but by definition the coverage can never exceed $b$. For the same reason, the cost of the projects we have added to $W$ cannot exceed $r$, and thus $c(W') \le c(W) + r$.
%
\end{proof}



With this lemma in hand, we now turn to the problem of finding high-coverage outcomes with good welfare.
Let $B \ge 0$ be a lower bound on the social welfare we desire. 
We rephrase our problem as an optimization problem in which we maximize the coverage of an outcome subject to a linear knapsack constraint and a linear covering constraint. The knapsack constraint enforces budget feasibility, and the covering constraint encodes the requirement that the total utility of the outcome is at least $B$.
%
\begin{equation}\label{P:DF1P}
\tag{DF1P}
\begin{split}
\max_{W \subseteq \mcP} & \,\, \cov(W)  \\
\text{s.t. }
& \sw(W) \ge B, \\
& \:\:\:c(W) \le b.
\end{split}
\end{equation}

The main tool we apply is a theorem on the maximization of nondecreasing submodular functions of \citet{mizrachi2018tight}. Recall that a set function is nondecreasing if its value never decreases as elements are added to its input, and submodular if it exhibits diminishing returns. 

\begin{definition}
    Given a finite set $\Omega$, a set function $f : \Omega \to \mathbb{R}_{\geq 0}$ is \emph{nondecreasing} and \emph{submodular} if for every $A, B \subseteq \Omega$ such that $A \subseteq B$ we have $f(A) \le f(B)$ and $f(A \cup \{x\}) - f(A) \ge f(B \cup \{x\}) - f(B)$ for all $x \in \Omega \setminus B$.
\end{definition}
The theorem we apply is as follows. 
\begin{theorem}[\citealp{mizrachi2018tight}, Theorem 5]
	\label{thm:mizrachi}
	For each constant $\epsilon > 0$, there exists a deterministic algorithm for maximizing a nondecreasing submodular function subject to one packing constraint and one covering constraint that runs in time $O(n^{O(1)})$, where $n = |\Omega|$ is the size of the support of the set function, satisfies the covering constraint up to a factor of $1 - \epsilon$ and the packing constraint up to a factor of $1+\epsilon$, and achieves an approximation ratio of $0.353$.
\end{theorem}

We apply this theorem to find a solution that satisfies a $0.353$ fraction of coverage and achieves optimal fairness-constrained utility. Then, we apply \Cref{lem:df1completion} to augment our solution using an additional $1 - 0.353 + \epsilon$ fraction of our budget in order to obtain a final solution which satisfies full DF1.
However, in order to apply \Cref{thm:mizrachi}, we must first establish that $\cov(W)$ is a nondecreasing submodular function. In particular, note that the coverage functions $\cov_i(W)$ for each district are clearly nondecreasing and submodular. It follows that their sum, $\cov(W)$ is also nondecreasing and submodular, yielding the following lemma.

\begin{lemma}
	\label{lem:submod}
	The function $\cov(W)$ is nondecreasing and submodular.
\end{lemma}

We are now ready to prove \Cref{thm:DF1det}, which applies the DF1 completion lemma to an approximately optimal solution for the problem \ref{P:DF1P}.

\begin{proof}[Proof of \Cref{thm:DF1det}]

Recall that we have assumed that the maximum utility of an outcome is polynomially bounded in $m$ and $k$ and that the maximum utility is integral. Thus, the value of $\OPT$ falls in a polynomial range. For each value $B$ in this range, solve the problem \ref{P:DF1P} using the algorithm from \Cref{thm:mizrachi}. Now consider all values of $B$ for which the algorithm returned a solution with $\cov(W) \ge 0.353b$; such a value must exist since we are guaranteed this condition when $B = \OPT$ (since for this value, the optimum of problem \eqref{P:DF1P} is $b$). Among all solutions we found that satisfy $\cov(W) \ge 0.353b$, take the one that maximizes $\sw(W)$. This solution provides social welfare \textit{at least} $(1 -\epsilon) \OPT$.

We have obtained an outcome $W$ with 
\[
\cov(W) \ge 0.353b = b - 0.647b,
\] 
and $\sw(W) \ge (1-\epsilon)\OPT$ and $c(W) \le (1+\epsilon) b$. Now apply \Cref{lem:df1completion} to $W$ to obtain a DF1 outcome $W' \supseteq W$ with
\[
c(W') \le c(W) + 0.647b \le (1+ 0.647 + \epsilon)b.
\]
This outcome $W'$ satisfies the requirements of \Cref{thm:DF1det}.
\end{proof}

\section{Discussion}

Our results extend to the special case of unit costs, also known as committee selection. In committee selection, we elect a committee to represent voters in a larger governmental body such as a parliament. Often, to ensure local representation, the electorate is split into voting districts, which elect their representatives separately. The districts may be apportioned different numbers of representatives, for example based on district size. While this scheme guarantees each district representation, it may well be possible to increase the welfare of the voters in a district, for example by electing a diverse array of candidates with expertise in various areas who can gather votes from across the electorate. Thus, it is natural for all districts to elect the committee together if we impose district-fairness constraints. This way, we can maximize social welfare of the final committee while guaranteeing each district fair representation. This gives a more holistic view of committee selection in exactly the same way we addressed participatory budgeting, only instead of pooling the budget between districts, we now pool seats on a committee.

Our model implicitly treats districts as atoms, and so district fairness is a kind of \emph{individual rationality} property. In turn, individual rationality is a type of strategyproofness: it incentivizes districts not to leave the central election and instead hold a separate one. Is it possible to design a voting scheme that is fully strategyproof for districts, so that districts do not have incentives to misreport the utilities of their residents? Unfortunately not: \citet{Pete18a} proves an impossibility theorem about committee elections which implies that there does not exist a voting rule that is efficient, district-fair, and also strategyproof. This result holds even for approval votes.

Several open questions remain. Most obvious is the question of whether can we achieve welfare maximization and DF1 in polynomial time while guaranteeing to overspend the budget by less than $1/e$. More broadly, it would be interesting to study our problem with more general utility functions such as submodular or even general monotone valuation functions.
Additionally, it would be exciting to study approximation algorithms which promise full district fairness. In Appendix B in the supplementary material, we present an algorithm which satisfies district fairness and provides a $1/2$-approximation to optimal district-fair social welfare in the special case of unanimous districts; it would be interesting to extend this result to the general case.


\bibliography{district-fair}

\begin{thebibliography}{30}
\providecommand{\natexlab}[1]{#1}
\providecommand{\url}[1]{\texttt{#1}}
\providecommand{\urlprefix}{URL }
\expandafter\ifx\csname urlstyle\endcsname\relax
  \providecommand{\doi}[1]{doi:\discretionary{}{}{}#1}\else
  \providecommand{\doi}{doi:\discretionary{}{}{}\begingroup
  \urlstyle{rm}\Url}\fi

\bibitem[{Arora, Hazan, and Kale(2012)}]{arora2012multiplicative}
Arora, S.; Hazan, E.; and Kale, S. 2012.
\newblock The multiplicative weights update method: a meta-algorithm and
  applications.
\newblock \emph{Theory of Computing} 8(1): 121--164.

\bibitem[{Aziz et~al.(2017)Aziz, Brill, Conitzer, Elkind, Freeman, and
  Walsh}]{justifiedRepresentation}
Aziz, H.; Brill, M.; Conitzer, V.; Elkind, E.; Freeman, R.; and Walsh, T. 2017.
\newblock Justified Representation in Approval-Based Committee Voting.
\newblock \emph{Social Choice and Welfare} 48(2): 461--485.

\bibitem[{Aziz et~al.(2018)Aziz, Elkind, Huang, Lackner,
  S{\'a}nchez-Fern{\'a}ndez, and Skowron}]{AEHLSS18}
Aziz, H.; Elkind, E.; Huang, S.; Lackner, M.; S{\'a}nchez-Fern{\'a}ndez, L.;
  and Skowron, P. 2018.
\newblock On the Complexity of Extended and Proportional Justified
  Representation.
\newblock In \emph{Proceedings of the 32nd AAAI Conference on Artificial
  Intelligence (AAAI)}, 902--909.

\bibitem[{Aziz and Shah(2020)}]{aziz2020participatory}
Aziz, H.; and Shah, N. 2020.
\newblock Participatory Budgeting: Models and Approaches.
\newblock In Rudas, T.; and P\'eli, G., eds., \emph{Pathways Between Social
  Science and Computational Social Science: Theories, Methods, and
  Interpretations}. Springer.

\bibitem[{Brill et~al.(2017)Brill, Freeman, Janson, and
  Lackner}]{BrillFJL17phragmen}
Brill, M.; Freeman, R.; Janson, S.; and Lackner, M. 2017.
\newblock Phragm{\'e}n's Voting Methods and Justified Representation.
\newblock In \emph{Proceedings of the 31st AAAI Conference on Artificial
  Intelligence (AAAI)}, 406--413.

\bibitem[{Brown(1979)}]{brown1979knapsack}
Brown, J.~R. 1979.
\newblock The knapsack sharing problem.
\newblock \emph{Operations Research} 27(2): 341--355.

\bibitem[{Budish(2011)}]{budish2011combinatorial}
Budish, E. 2011.
\newblock The combinatorial assignment problem: Approximate competitive
  equilibrium from equal incomes.
\newblock \emph{Journal of Political Economy} 119(6): 1061--1103.

\bibitem[{Cabannes(2004)}]{Cabannes}
Cabannes, Y. 2004.
\newblock Participatory budgeting: a significant contribution to participatory
  democracy.
\newblock \emph{Environment and Urbanization} 16(1): 27--46.

\bibitem[{Cabannes(2014)}]{cabannes2014contribution}
Cabannes, Y. 2014.
\newblock Contribution of Participatory Budgeting to provision and management
  of basic services.
\newblock \emph{London: IIED} .

\bibitem[{Cabannes(2017)}]{cabannes2017participatory}
Cabannes, Y. 2017.
\newblock Participatory budgeting in Paris: Act, reflect, grow.
\newblock \emph{Another city is possible with participatory budgeting}
  179--203.

\bibitem[{Chekuri and Khanna(2005)}]{chekuri2005polynomial}
Chekuri, C.; and Khanna, S. 2005.
\newblock A polynomial time approximation scheme for the multiple knapsack
  problem.
\newblock \emph{SIAM Journal on Computing} 35(3): 713--728.

\bibitem[{Cheng et~al.(2019)Cheng, Jiang, Munagala, and Wang}]{cheng2019group}
Cheng, Y.; Jiang, Z.; Munagala, K.; and Wang, K. 2019.
\newblock Group fairness in committee selection.
\newblock In \emph{Proceedings of the 20th ACM Conference on Economics and
  Computation (ACM EC)}, 263--279.

\bibitem[{Fain, Goel, and Munagala(2016)}]{FGM16a}
Fain, B.; Goel, A.; and Munagala, K. 2016.
\newblock The core of the participatory budgeting problem.
\newblock In \emph{Proceedings of the 12th International Conference on Web and
  Internet Economics (WINE)}, 384--399.

\bibitem[{Fain, Munagala, and Shah(2018)}]{FMS18}
Fain, B.; Munagala, K.; and Shah, N. 2018.
\newblock Fair Allocation of Indivisible Public Goods.
\newblock In \emph{Proceedings of the 19th ACM Conference on Economics and
  Computation (ACM EC)}, 575--592.
\newblock Extended version arXiv:1805.03164.

\bibitem[{Faliszewski et~al.(2017)Faliszewski, Skowron, Slinko, and
  Talmon}]{FSST17a}
Faliszewski, P.; Skowron, P.; Slinko, A.; and Talmon, N. 2017.
\newblock Multiwinner Voting: A New Challenge for Social Choice Theory.
\newblock In Endriss, U., ed., \emph{Trends in Computational Social Choice},
  chapter~2.

\bibitem[{Fujimoto and Yamada(2006)}]{fujimoto2006exact}
Fujimoto, M.; and Yamada, T. 2006.
\newblock An exact algorithm for the knapsack sharing problem with common
  items.
\newblock \emph{European Journal of Operational Research} 171(2): 693--707.

\bibitem[{Garey and Johnson(1979)}]{garey1979computers}
Garey, M.~R.; and Johnson, D.~S. 1979.
\newblock \emph{Computers and intractability}, volume 174.
\newblock Freeman San Francisco.

\bibitem[{Goel et~al.(2019)Goel, Krishnaswamy, Sakshuwong, and
  Aitamurto}]{goel2019knapsack}
Goel, A.; Krishnaswamy, A.~K.; Sakshuwong, S.; and Aitamurto, T. 2019.
\newblock Knapsack voting for participatory budgeting.
\newblock \emph{ACM Transactions on Economics and Computation (TEAC)} 7(2):
  1--27.

\bibitem[{Hifi, M'Halla, and Sadfi(2005)}]{hifi2005exact}
Hifi, M.; M'Halla, H.; and Sadfi, S. 2005.
\newblock An exact algorithm for the knapsack sharing problem.
\newblock \emph{Computers \& Operations Research} 32(5): 1311--1324.

\bibitem[{Hoeffding(1963)}]{Hoeff63}
Hoeffding, W. 1963.
\newblock Probability inequalities for sums of bounded random variables.
\newblock \emph{Journal of the American Statistical Association} 58(301):
  13--30.

\bibitem[{Janson(2016)}]{Janson16arxiv}
Janson, S. 2016.
\newblock Phragm{\'{e}}n's and {T}hiele's election methods.
\newblock Technical report.
\newblock ArXiv:1611.08826 [math.HO].

\bibitem[{Jiang, Munagala, and Wang(2020)}]{jiang2020approximately}
Jiang, Z.; Munagala, K.; and Wang, K. 2020.
\newblock Approximately stable committee selection.
\newblock In \emph{Proceedings of the 52nd Annual ACM SIGACT Symposium on
  Theory of Computing (STOC)}, 463--472.

\bibitem[{Mizrachi et~al.(2018)Mizrachi, Schwartz, Spoerhase, and
  Uniyal}]{mizrachi2018tight}
Mizrachi, E.; Schwartz, R.; Spoerhase, J.; and Uniyal, S. 2018.
\newblock A tight approximation for submodular maximization with mixed packing
  and covering constraints.
\newblock arXiv:1804.10947.

\bibitem[{Peters(2018)}]{Pete18a}
Peters, D. 2018.
\newblock Proportionality and Strategyproofness in Multiwinner Elections.
\newblock In \emph{Proceedings of the 17th International Conference on
  Autonomous Agents and Multiagent Systems (AAMAS)}, volume 1549--1557.

\bibitem[{Peters and Skowron(2020)}]{peters2020proportionality}
Peters, D.; and Skowron, P. 2020.
\newblock Proportionality and the limits of welfarism.
\newblock In \emph{Proceedings of the 21st ACM Conference on Economics and
  Computation (ACM EC)}, 793--794.

\bibitem[{S{\'a}nchez-Fern{\'a}ndez et~al.(2017)S{\'a}nchez-Fern{\'a}ndez,
  Elkind, Lackner, Fern{\'a}ndez, Fisteus, {Basanta Val}, and Skowron}]{pjr17}
S{\'a}nchez-Fern{\'a}ndez, L.; Elkind, E.; Lackner, M.; Fern{\'a}ndez, N.;
  Fisteus, J.~A.; {Basanta Val}, P.; and Skowron, P. 2017.
\newblock Proportional justified representation.
\newblock In \emph{Proceedings of the 31st AAAI Conference on Artificial
  Intelligence (AAAI)}, 670--676.

\bibitem[{Talmon and Faliszewski(2019)}]{talmon2019framework}
Talmon, N.; and Faliszewski, P. 2019.
\newblock A framework for approval-based budgeting methods.
\newblock In \emph{Proceedings of the 33rd Conference on Artificial
  Intelligence (AAAI)}, volume~33, 2181--2188.

\bibitem[{Williamson and Shmoys(2011)}]{williamson2011design}
Williamson, D.~P.; and Shmoys, D.~B. 2011.
\newblock \emph{The design of approximation algorithms}.
\newblock Cambridge University Press.

\bibitem[{Yamada and Futakawa(1997)}]{yamada1997heuristic}
Yamada, T.; and Futakawa, M. 1997.
\newblock Heuristic and reduction algorithms for the knapsack sharing problem.
\newblock \emph{Computers \& operations research} 24(10): 961--967.

\bibitem[{Yamada, Futakawa, and Kataoka(1998)}]{yamada1998some}
Yamada, T.; Futakawa, M.; and Kataoka, S. 1998.
\newblock Some exact algorithms for the knapsack sharing problem.
\newblock \emph{European Journal of Operational Research} 106(1): 177--183.

\end{thebibliography}


\appendix

\section{$\frac{1}{2}$-Approximation for Unanimous Districts with Unit Costs}
\label{app:halfapx}

In this section we study approximation algorithms for the simplest version of our problem which we know to be NP-hard: when each districts consist of a single voter and every project has unit cost. In fact, we will study a strictly more general setting than each district consisting of a single voter; namely, we study the setting where each district is ``unanimous.'' Formally, we study instances of district-fair welfare maximization where $c(x_j) = 1$ for all $j$ and $\sw_i(x_j) \in \{0, |d_i|\}$ for all $i$ where $|d_i|$ is the number of voters in $d_i$. For this setting we will give a $\frac{1}{2}$-approximation.

Our algorithm will make use of the following notion of conditional coverage which builds on \Cref{dfn:covi}.
\begin{definition}
	The coverage of a project $x_j$ given an outcome $W$ is $\cov(x_j | W) \coloneqq \sum_i \cov_i(W \cup \{x_j\}) - \cov_i(W)$.
\end{definition}
Notice that by our assumption of unit cost and unanimous districts we have that $\cov(x_j | W) \in \mathbb{Z}_{\geq 0}$.



We now present our greedy algorithm that satisfies district fairness and achieves a $1/2$-approximation to the optimal district-fair utility for the setting of unanimous districts and unit costs. Formally, the algorithm, which we call the \emph{Unanimous Greedy Algorithm (UGA)} proceeds as follows.

\begin{enumerate}
	\item Given an instance $\mcI$, initialize $W_0 \gets \emptyset$.
	\item For $j \in [b]$:
	\begin{enumerate}
		\item Let $c_j := \max_{x_j} \cov(x_j|W_{j-1})$ be the max possible coverage and let $X_j := \{x_j : \cov(x_j| W_{j-1}) = c_j\}$ be all projects which achieve this coverage.
		\item Let $x_j \coloneqq \argmax_{x_j \in X_j} \sw(x_j)$ be the max covering project with maximum utility.
		\item Update $W_j \gets W_{j-1} \cup \{x_j\}$.
	\end{enumerate}
	\item Return $W_b$.
\end{enumerate}


\begin{theorem}
	\label{thm:UGA}
	Given an instance $\mcI$ consisting of unanimous approval districts, UGA returns a solution which satisfies district fairness and achieves a $1/2$-approximation to the optimal district-fair utility. 
\end{theorem}

\begin{proof}
	
	Let $W$ represent the result of UGA, and let $W^*$ represent the optimal district-fair outcome. Furthermore, let $N_2, N_1$ and $N_0$ be all projects purchased by UGA that had conditional coverage at least 2, exactly 1 and exactly 0 when purchased by UGA respectively. Clearly $W$ is district-fair and budget-feasible and so we need only argue that it achieves at least $\sw(W^*)/ 2$ utility.
	
	Now, consider the following subproblem, which we will call $\mcI'$, which intuitively is our original instance $\mcI$ but where all of $N_2$ is forced to be in a solution and no projects from $N_0$ are available. More formally, $\mcI'$ is $\mcI$ but where our budget is changed to $b' := b-|N_2| - |N_0|$, $f_i$ is changed to $f_i' := \max(f_i - \sw_i(N_2), 0)$ for all $i$ and the set of purchasable projects is $\mcP' := \mcP \setminus N_2$. $\sw_i(x_j)$ is the same for all $i,j$ in $\mcI'$ as in $\mcI$. Notice that the coverage of any project in $\mcI'$ is at most $1$ but the total coverage required for fairness is $b'$, meaning that every budget-feasible and district-fair solution for $\mcI'$ has size exactly $b'$. Also notice that $N_1$ is not only feasible for $\mcI'$ but also attains the optimal utility among all district-fair and budget-feasible solutions.
	
	
	We claim that there exists a subset $N_1^* \subseteq W^* \setminus N_2$ which is district-fair and budget-feasible for $\mcI'$. To see this, note that we  can iteratively build $N_1^*$ by initializing it to $\emptyset$ and then repeatedly adding to it any $x_j \in W^* \setminus N_1^* \setminus N_2$ such that $\cov(x_j | N_1^*)$ in $\mcI'$ is at least $1$. After $b'$ such additions we are guaranteed to have a district-fair and budget-feasible solution for $\mcI'$ and such an $x_j$ always exists since $W^* \setminus N_2$ is district-fair for $\mcI'$. As noted above, any district-fair and budget-feasible solution for $\mcI'$ has size $b'$ and so $|N_1^*| = b'$.
	
	
	Thus, since $N_1$ is optimal for $\mcI'$, we know that $N_1$ achieves at least as high utility as $N_1^*$, i.e., $|N_1| = |N_1^*| = b'$, and
	\begin{align}\label{eq:firstbound}
	\sw(N_1^*) \leq sw(N_1).
	\end{align}
	
	It remains to understand the utility of $W^* \setminus N_1^*$. However, note that at least half of the projects other than $N_1$ must be in the $N_0$ phase. That is, $|N_0| \geq |N_2| / 2$. Intuitively, this means that UGA ``frees up'' at least $\frac{b - b'}{2}$ money to spend on high-utility projects. Let $\mcP_0 := \mcP \setminus \{ N_2 \cup N_1\}$ be all projects not in $N_2$ or $N_1$. We have that $\sw(N_0)$ is the utility of the top $\frac{b - b'}{2}$ projects in $\mcP_0$. On the other hand, consider projects in $W^* \setminus N_1^*$. We can divide these into projects which are in $N_2$ and $N_1$ and those which are not. In particular, let $W^*_{1,2} := W^* \setminus N_1^* \cap (N_1 \cup N_2)$ and let $W^*_{0} := W^* \setminus N_1^* \cap \mcP_0$ so that $W^*_{1,2} \cup W^*_0 = W^* \setminus N_1^*$. Now notice that trivially
	\begin{align}\label{eq:secondbound}
	\sw(W^*_{1,2}) \leq \sw(N_1) + \sw(N_2).
	\end{align}
	On the other hand, $|W^*_0| = b - b' - |W^*_{1,2}| \leq b - b'$ and $W^*_0 \subseteq \mcP_0$ and so $\sw(W^*_0)$ is at most the utility of the $b - b'$ highest utility projects in $\mcP_0$. Since our utilities are additive and $\sw(N_0)$ is the utility of the $\frac{b-b'}{2}$ highest utility projects in $\mcP_0$, it follows that
	\begin{align}\label{eq:thirdbound}
	\sw(W_{0}^*) \leq 2 \cdot \sw(N_0).
	\end{align}
	
	Since $W^* = N_1^* \cup W_0^* \cup W_{1,2}^*$, we can combine the above bounds to conclude our $\frac{1}{2}$-approximation. Namely, applying the additivity of our utilities and combining Equations \ref{eq:firstbound}, \ref{eq:secondbound} and \ref{eq:thirdbound} we have
	\begin{align*}
	\sw(W^*) &= \sw(N_1^*) + \sw(W_0^*) + \sw (W_{1,2}^*)\\
	&\leq 2 \cdot \sw(N_0) + 2 \cdot \sw(N_1) + \sw(N_2)\\
	& \leq 2 \cdot \sw(W)
	\end{align*}
	and so we conclude that $\sw(W) \geq \frac{\sw(W^*)}{2}$.
	
	
\end{proof}

\section{Integrality Gap}
\label{app:integralitygap}

Here, we investigate the integrality gap of the natural LP for our problem. As a reminder, the integrality gap of an LP measures how much better a fractional solution can do than an integral solution. An unbounded integrality gap shows that any analysis of an approximation algorithm which charges the value of its integral solution to the value of the optimal LP gives an unboundedely-bad approximation ratio. For this reason integrality gaps are sometimes taken as evidence of hardness of approximation. For more details on the topic of integrality gaps see \citet{williamson2011design}. 

We will show that our LP has an unbounded integrality gap which suggests that approximation algorithms which return budget-feasible and district-fair solutions with nearly-optimal social welfare may be difficult or impossible to attain for the general case. 

Formally our LP and its integrality gap are as follows. Our LP has a variable $y_j$ for each project $x_j$ corresponding to the extent to which we choose $x_j$.
\begin{equation}\label{LP:DFLP}\tag{DFLP}
\begin{split}
\max & \sum_{j}  y_j \cdot \sw(x_j) \\ 
\text{s.t. } & \sum_j y_j \cdot c(x_j) \leq b\\
& \sum_j y_j \cdot \sw_i(x_j) \geq f_i \qquad \forall i\\
& 0 \leq y_j \leq 1 \qquad \forall j
\end{split}
\end{equation}
We let $\ref{LP:DFLP}(\mcI)$ correspond to the polytope corresponding to the above LP for an instance $\mcI$ of district-fair welfare maximization.

The integrality gap of \ref{LP:DFLP} is defined as
\begin{align*}
\min_{\mcI} \frac{\max_{y \in \ref{LP:DFLP}(\mcI)\cap \mathbb{Z}^m}\sum_{j} y_j \cdot \sw(x_j)}{\max_{y \in \ref{LP:DFLP}(\mcI)}\sum_{j} y_j \cdot \sw(x_j)}.
\end{align*}

The basic idea of our integrality gap construction is as follows. We will construct an instance of social-welfare maximization where the preferences of each district are ``circular''. In particular, each district will like  two projects and every project will be liked by exactly two districts. As in our NP-hardness proof, we will also have a collection of dummy projects which are given very high utility by dummy districts which deserve no utility. An optimal fractional solution will be able to choose each non-dummy project to extent essentially $\frac{1}{2}$ to satisfy district-fairness and then spend its remaining budget on high-utility dummy projects. On the other hand, the optimal integral solution will have to spend its entire budget satisfying fairness.

\begin{restatable}{theorem}{IGLB}
	There does not exist a function $f$ such that the integrality gap of \ref{LP:DFLP} is at most $\frac{1}{f(k,m)}$. Further, this integrality gap holds even when all projects have unit cost.
\end{restatable}  
\begin{proof}
	Fix $k \in \mathbb{Z}_{\geq 1}$ and a sufficiently small $\epsilon > 0$. We define our instance of social-welfare maximization on $k$ districts where $d_1, d_2, \ldots d_{k-1}$ will be non-dummy districts and the district $d_k$ will be a dummy district. Similarly, we will have $2(k-1)$ projects where $x_1, x_2, \ldots, x_{k-1}$ will be non-dummy projects and the remaining projects $x_{k}, \ldots, x_{2(k-1)}$ will be dummy projects.
	
	For each non-dummy district $d_i$ we let $b_i = 1$ and define its utility for project $x_j$ as
	\begin{align*}
	\sw_i(x_j) := 
	\begin{cases}
	1 + \epsilon &\text{if $j = i$}\\
	1 & \text{if $j = (i+ 1 \text{ mod } k-1) + 1$}\\
	0 &\text{otherwise}
	\end{cases}
	\end{align*}
	
	For the dummy district $d_k$  we let $b_k=0$ and define its utility for project $x_j$ as 
	\begin{align*}
	\sw_k(x_j) := 
	\begin{cases}
	B &\text{if $x_j$ is a dummy project}\\
	0 & \text{otherwise}
	\end{cases}
	\end{align*}
	for $B$ sufficiently large to be chosen later. Notice that $\sw(x_j) = B$ for each dummy project $x_j$. Lastly, we let our budget $b = k-1$ and we let $c(x_j) = 1$ for all $x_j$.
	
	Now notice that each non-dummy district $d_i$ has $f_i = 1 + \epsilon$. Consequently, any district-fair integral solution must include all non-dummy projects, namely $x_1, x_2, \ldots, x_{k-1}$. However, since $b = k-1$, it follows that the only district-fair integral solution is $W_{int} := \{x_1, x_2, \ldots, x_{k-1}\}$ where $\sw(W_{int}) = (2+\epsilon)(k-1)$.
	
	On the other hand, consider the following fractional solution $y$. For each non-dummy project $x_j$ we let $y_j = \frac{1+\epsilon}{2}$. For each dummy project $x_j$ we let $y_j$ be $(1-\frac{1+\epsilon}{2})$. Clearly $\sum_{j} y_j \leq b$. Moreover, notice that for each district $d_i$ we have $\sum_j y_j \cdot \sw_i(x_j) = \frac{1+\epsilon}{2}(2 + \epsilon) \geq 1+\epsilon = f_i$ and so our solution is indeed in the polytope of \ref{LP:DFLP}. However, since $y_j = (1-\frac{1+\epsilon}{2})$ for each dummy project we have that $\sum_{j} y_j \cdot \sw(x_j) \geq B(k-1)(1-\frac{1+\epsilon}{2})$.
	
	Thus, for the above instance we have that the ratio of the optimal integral solution to the optimal fractional solution is at most
	\begin{align*}
	\frac{(2+\epsilon)(k-1)}{B(k-1)\left(1 - \frac{1+\epsilon}{2}\right)} \leq \frac{10}{B}.
	\end{align*}
	
	Since $B$ can be chosen independently of $k$ and $m$, we have that the above instance has integrality gap strictly less than $\frac{1}{f(k, m)}$ for any function $f$ of $k$ and $m$.
\end{proof}

We note that the proof of the above result also rules out any integrality gap which is  which is larger than $o(\frac{1}{c})$ where $c$ is the total number of voters across all districts.

\section{Randomized Optimal DF1 Outcome with Extra Budget}\label{sec:randDF1}

In this section we give our randomized analogues of \Cref{thm:DF1det} and \Cref{cor:det}. We use the notation of \Cref{sec:DF1} throughout this section. Whereas our deterministic algorithms overspend budget by $64.7 \%$, our randomized algorithms will only overspend it by $\frac{1}{e}$ with high probability. Formally, we show the following theorem.

\begin{restatable}{theorem}{Fone}
    \label{thm:DF1}
    There is a $\poly(m,k)$-time algorithm which, given an instance of district-fair welfare maximization, returns an outcome $W$ such that $W$ is DF1, $c(W) \leq \left(1 + \frac{1}{e} + \epsilon\right) b \approx 1.37 b$ with high probability, and $\sw(W) \ge (1 - \epsilon) \OPT$ for any fixed constant $\eps > 0$.
\end{restatable}

As with \Cref{cor:det} for \Cref{thm:DF1det}, we immediately have a corollary which gives an algorithm which does not overspend its budget.

\begin{corollary}
    There is a $\poly(m,k)$-time algorithm which, given an instance of district-fair welfare maximization, returns an outcome $W$ such that $W$ is DF1 for $\mcI'\left(\frac{1}{1 + 1/e}\right)$, $c(W) \leq \left(1 + \epsilon\right) b$ with high probability, and $\sw(W) \ge  (1 - \epsilon) \OPT'\left(\frac{1}{1 + 1/e}\right)$ for any fixed constant $\eps > 0$.
\end{corollary}

On a high level, this proof will closely follow that of \Cref{thm:DF1det}. In particular, it uses the same submodular optimization framing of the problem (i.e., \ref{P:DF1P}). However, there are some notable differences. In particular, we leverage the following randomized result from \citet{mizrachi2018tight} instead of the previous deterministic result from \citet{mizrachi2018tight}.

\begin{theorem}[\citealp{mizrachi2018tight}, Theorem 1]
	\label{thm:mizrachirand}
	For each constant $\epsilon > 0$, there exists a randomized algorithm for maximizing a nondecreasing submodular function subject to one packing constraint and one covering constraint that runs in time $O(|\Omega|^{O(1)})$, satisfies the packing constraint, satisfies the covering constraint up to a factor of $1 - \epsilon$, and achieves an expected approximation ratio of $1 - \frac{1}{e} - \epsilon$.
\end{theorem}


Additionally, we require Hoeffding's inequality~\cite{Hoeff63}, which bounds the probability that the sum of a sequence of independent random variables deviates from its expectation; we restate it below.
\begin{theorem}[Hoeffding's Inequality]
	Given a sequence of $n$ independent random variables $Y_1, \dots, Y_n$, where each $Y_i$ takes a value in the range $[\alpha_i, \beta_i]$, we have that 
	\[\Pr\left( \E\left[\sum_{k=1}^n Y_k\right] - \sum_{k=1}^n Y_k \ge t \right) \le e^{-\frac{2t^2}{\sum_{i=1}^n (\beta_i - \alpha_i)^2}},\]
	where $t \ge 0$.
\end{theorem}

We now present the main lemma of the section, which states that there exists a polynomial-time algorithm which, given a guess of the optimal district-fair social welfare $B \le \OPT$, returns an outcome $W$ that is DF1 fair, overspends the budget by approximately $b/e$ with high probability, and achieves at least $(1 - \epsilon)B$ social welfare. 

\begin{lemma}
	\label{lem:DF1Pwhp}
	Given an instance of \ref{P:DF1P}, there is an algorithm that runs in polynomial time which returns an outcome $W$ such that $W$ is DF1, $c(W) \leq \left(1 + \frac{1}{e} + \epsilon \right) b$ with high probability, and $\sw(W) \ge (1 - \epsilon) B$,
	where $\epsilon > 0$ is an arbitrary constant and $B \le \OPT$.
\end{lemma}

\begin{proof}
	Note that, by \Cref{lem:df1completion}, it suffices to find an outcome $W'$ such that $c(W') \le b$, $\cov(W') \ge \left(1 - \frac{1}{e} - \epsilon\right) b$, and $\sw(W') \ge (1 - \epsilon) B$, as we can complete this solution into a DF1 solution with $\left( \frac{1}{e} + \epsilon\right) b$ more budget. 
	
	Let $\epsilon_0 \coloneqq \epsilon / 2$. By \Cref{lem:submod}, we may apply \Cref{thm:mizrachirand} in order to find a solution $W$ such that $\cov(W) \ge (1-\epsilon_0)(1 - 1/e)$ in expectation, $\sw(W) \ge B$, and $c(W) \le b$. We now show how to transform this guarantee in expectation into one with high probability by using Hoeffding's inequality and an averaging argument. 
	
	In order to apply Hoeffding's inequality to our setting, let $Y_j$ represent the coverage of the $j^{th}$ run of our application of \Cref{thm:mizrachirand}. We know that each run is independent, and therefore the $Y_j$'s are also independent.
	Let $n = \omega(\log k/\epsilon_0^2)$, and define a sequence of $n$ independent random variables $Y_1, \dots, Y_n$ representing the coverage of $n$ runs of the mechanism. Furthermore, because coverage is bounded between 0 and $b$, we have that $\beta_i = b$ and $\alpha_i = 0$ for all $i \in [n]$.
	
	By Hoeffding, we have
	\begin{align*}
	&\Pr\left( \E\left[\sum_{j=1}^n Y_j\right] - \sum_{j=1}^n Y_j \ge \epsilon_0 \left( 1 - \frac{1}{e}\right) bn\right) \\
	&\qquad\le \exp\left(-\frac{2(\epsilon_0 \left( 1 - \frac{1}{e}\right)bn)^2}{\sum_{i=1}^n (\beta_i - \alpha_i)^2}\right) \\
	&\qquad= \exp\left(-2\epsilon_0^2 \left( 1 - \frac{1}{e}\right)^2 n\right).
	\end{align*}
	
	Because we set $n = \omega(\log k/\epsilon_0^2)$, this probability goes to 0 polynomially quickly. Therefore, we know that, with high probability, $\E\left[\sum_{j=1}^n Y_j\right] - \sum_{j=1}^n Y_j \le \epsilon_0 \left( 1 - \frac{1}{e}\right)b n$, or $\sum_{j=1}^n Y_j \ge bn (1-\epsilon_0)\left( 1 - \frac{1}{e}\right) - \epsilon_0 \left( 1 - \frac{1}{e}\right) bn$. By an averaging argument, this means that, with high probability, there exists a $Y_i$ such that $Y_i \ge (1-\epsilon_0-\epsilon_0)\left( 1 - \frac{1}{e}\right)b = (1 - \epsilon)\left( 1 - \frac{1}{e}\right)b$, as desired.
\end{proof}

With this lemma in hand, we are ready to prove \Cref{thm:DF1}.

\begin{proof}[Proof of \Cref{thm:DF1}]
	Recall that we have assumed that the maximum utility of an outcome is polynomially bounded in $m$ and $k$ and that the maximum utility is integral. Thus, the value of $\OPT$ falls in a polynomial range. For each value $B$ in this range, run the procedure from \Cref{lem:DF1Pwhp}. By a union bound over these polynomially-many applications of \Cref{lem:DF1Pwhp}, we have that all applications of \Cref{lem:DF1Pwhp} with $B \le \OPT$ succeed with high probability, resulting in a collection of solutions; one for each $B$. Now consider all values of $B$ for which the algorithm returned a solution with $c(W) \le (1 + \frac{1}{e} + \eps)b$; such a value must exist since with high probability we are guaranteed this condition when $B = \OPT$ and we know that, with high probability, all applications of \Cref{lem:DF1Pwhp} succeed for all $B \le \OPT$. Among all solutions we found that satisfy $c(W) \le (1 + \frac{1}{e} + \eps)b$, take the one that maximizes $\sw(W)$. This solution provides utility \textit{at least} $\OPT$.

\end{proof}

\end{document}